\documentclass{article}
\usepackage[T1]{fontenc}
\usepackage[utf8]{inputenc}
\usepackage[english]{babel}
\usepackage{amsmath,amssymb, amsthm}

\newtheorem{theorem}{Theorem}

\newtheorem{lemma}{Lemma}

\newtheorem{proposition}{Proposition}

\title{Inapproximability of sufficient reasons for decision trees}

\author{Alexander Kozachinskiy$^{1,2,3}$}
\date{\small %
    $^1$Centro Nacional de Inteligencia Artificial, Chile\\%
    $^2$Instituto de Ingenier\'ia Matem\'atica y Computacional, Universidad Cat\'olica de Chile\\
    $^3$Instituto Milenio Fundamentos de los Datos, Chile}

\begin{document}

\maketitle

\begin{abstract}
    In this note, we establish  the hardness of approximation  of the problem of computing the  minimal size of a  $\delta$-sufficient reason for  decision trees.
\end{abstract}

\section{Introduction}

 Given input to an AI model, one might want to understand, why the model classifies this input  the way it does. \emph{Local explanations} like this have been extensively studied in recent years~\cite{shih2018symbolic,ribeiro2018anchors,lundberg2017unified,yan2021if}. One kind of them, called \emph{sufficient} reasons, are in the focus of this note.

For simplicity, by AI models we understand Boolean functions $\mathcal{M}\colon\{0,1\}^n\to\{0, 1\}$, thus restricting ourselves to binary classification of vectors of binary features. There are many ways $\mathcal{M}$ can be given, e.g., as a decision tree, or as a neural network, or even as a Boolean circuit.

A \emph{sufficient reason}~\cite{shih2018symbolic,ignatiev2021sat}
for an input $x = (x_1, \ldots, x_n)\in\{0, 1\}^n$ under the model $\mathcal{M}\colon\{0,1\}^n\to\{0,1\}$ is a set $S\subseteq\{1, \ldots, n\}$ of coordinates such that for every $x^\prime\in \{0,1\}^n$ that coincides with $x$ in all coordinates from $S$ it holds that $\mathcal{M}(x^\prime) = \mathcal{M}(x)$. In other words, a sufficient reason for $x$ is a set of its features such that all inputs with these features are classified as $x$.

Barcel{\'o} et al.~\cite{barcelo2020model} initiated the study of the \emph{computational complexity} of local explanations for AI models. In particular, they introduced the \emph{minimum sufficient reason} problem. This problem is: given an input $x\in\{0, 1\}^n$ and a model $\mathcal{M}\colon\{0, 1\}^n\to\{0,1 \}$, find a sufficient reason $S$ for $x$ under $\mathcal{M}$ of minimal size. They showed that even when $\mathcal{M}$ is given as a decision tree, this problem is NP-hard (and its decision version is NP-complete).

Recently, it was noticed~\cite{arenas2022computing,waldchen2021computational} that sufficient reasons as explanations are too restrictive, and are usually too big in practice. We might be happy with a smaller set of features which guaranties the same classification as for our input $x$ for most of the cases (possibly, not always). This leads to the following relaxation of the notion of a sufficient reason. Namely, take any $\delta\in [0, 1]$. A $\delta$-sufficient reason for $x\in\{0, 1\}^n$ under $\mathcal{M}\colon\{0,1\}^n\to\{0, 1\}$ is a set $S\subseteq \{1, 2, \ldots, n\}$ of coordinates for which the following holds: if we consider the set of all $x^\prime\in\{0, 1\}^n$ that coincide with $x$ in all coordinates from $S$, then for at least  a fraction $\delta$ of these $x^\prime$ we have $\mathcal{M}(x^\prime) = \mathcal{M}(x)$. 

Just as with ``exact'' sufficient reasons, one can consider the minimum $\delta$-sufficient reason problem. In this problem, given an input $x\in\{0, 1\}^n$ and a model $\mathcal{M}\colon\{0, 1\}^n\to\{0, 1\}$, our task is to find a $\delta$-sufficient reason $S$ for $x$ under $\mathcal{M}$ of minimal size. We assume that $\delta$ is a parameter of the problem and not a part of the input (thus, for different $\delta\in[0, 1]$ we get different algorithmic problems).

W{\"a}ldchen et al.~\cite{waldchen2021computational} show that when $\mathcal{M}$ is given as a Boolean circuit, the decision version of the minimum $\delta$-sufficient problem is $NP^{PP}$-complete\footnote{Thus, when $\mathcal{M}$ is given by a circuit, the problem even jumps out of $NP$ (unless, of course, $NP = NP^{PP}$). The reason is that verifying if a set of coordinates is  a $\delta$-sufficient reason requires checking, whether at least a fraction $\delta$ of inputs to some circuit are satisfying.}, for any $\delta\in [\frac 1 2, 1)$. In turn, Arenas et al.~\cite{arenas2022computing} show that the decision version of the minimum $\delta$-sufficient reason problem is NP-complete for every $\delta\in (0, 1)$ when $\mathcal{M}$ is given as a decision tree.

Given that these problems are hard to solve exactly, it makes sense to investigate their hardness of approximation. W{\"a}ldchen et al.~\cite{waldchen2021computational} establish strong hardness of approximation for the minimum $\delta$-sufficient reason problem when $\mathcal{M}$ is given by a Boolean circuit. Namely, they show that it is NP-hard to distinguish the case when there exists a ``very good'' and ``very small'' sufficient reason from the case when all ``mildly good'' sufficient reasons include  ``almost all'' coordinates. More specifically, for every $\varepsilon > 0$, they show that it is NP-hard to distinguish the case when there exists a $(1 - \varepsilon)$-sufficient reason of size at most $n^\varepsilon$ from the case when there is no $\varepsilon$-sufficient reason of size less then $n - n^\varepsilon$. Here $n$ denotes the number of variables. 

We show that the same hardness of approximation holds even when $\mathcal{M}$ is given by a decision tree. This improves upon the result of W{\"a}ldchen et al.~because every decision tree can be converted in polynomial time into a Boolean circuit, computing the same function (thus, the problem for decision trees can only be easier than for Boolean circuits). However, our complexity assumption is slightly stronger than $P\neq NP$.

\begin{theorem}
\label{minimum}
Unless SAT can be solved in quasi-polynomial time, for every $\varepsilon > 0$ there exists no polynomial-time algorithm that, given a decision tree $T$ over $n$ variables and an input $x\in\{0,1\}^n$ to it, distinguishes between the following two cases:
\begin{itemize}
\item there exists a $(1 - \varepsilon)$-sufficient reason for $x$ under $T$ of size at most $n^{\varepsilon}$
\item  there exists no $\varepsilon$-sufficient reason for $x$ under $T$ of size less than $n - n^{\varepsilon}$.
\end{itemize}
\end{theorem}
The rest of the note is organized as follows.
In Section \ref{prelim}, we give preliminaries. In Section \ref{inter}, we introduce an intermediate problem, and show that its hardness implies Theorem \ref{minimum}. In the last section, we establish the hardness of the intermediate problem.

\section{Preliminaries}
\label{prelim}

A \emph{decision tree} $T$ over $n$ variables is a binary rooted tree, where
\begin{itemize}
    \item every leaf is labeled by $0$ or $1$;
    \item every inner node $v$ has exactly two children; one of them is labeled as the $0$-child of $v$, and the other one is labeled as the $1$-child of $v$.
    \item in addition, every inner node is labeled by an index $i\in\{1, \ldots, n\}$.
\end{itemize}
Given an input $x = (x_1, \ldots, x_n)\in\{0, 1\}^n$, the value of $T$ on $x$, denoted by $T(x)$, is computed as follows. Starting at the root of $T$, we descend to one of the leafs. More specifically, if we are in an inner node $v$, we take its label $i\in\{1, \ldots, n\}$, and if $x_i = 0$, we descend to the $0$-child of $v$, and if $x_i = 1$, we descend to the $1$-child of $v$. When we reach a leaf, we let $T(x)$ be the label of this leaf.

Along with inputs from $\{0, 1\}^n$ (``complete'' inputs), we will consider \emph{partial} inputs to $T$. Formally, a partial input is an element of $\{0, 1, \bot\}^n$, where $\bot$ is a special symbol meaning ``undefined''.  We say that two partial inputs $x, y \in\{0, 1, \bot\}^n$ are \emph{consistent} if there exists no $i\in \{1, \ldots, n\}$ such that $x_i \neq y_i$ and $x_i, y_i\in\{0, 1\}$. This means that $x$ and $y$ can be ``completed'' to the same complete input by changing undefined coordinates to 0 or 1.

We extend the notation $T(x)$ to partial inputs. Namely, for a  partial input $y\in\{0, 1, \bot\}^n$ we define:
\[T(y) = \frac{\Big|\{x\in\{0,1\}^n \mid \text{$x$ and $y$ are consistent and } T(x) = 1\}\Big|}{\Big|\{x\in\{0,1\}^n \mid \text{$x$ and $y$ are consistent }\}\Big|}.\]

In the denominator, we have the number of \emph{completions} of $y$, that is, the number of complete inputs that are consistent with $y$. In the numerator, we have the number of completions of $y$ that satisfy $T(x) = 1$. Hence, $T(y)$ is equal to the probability that a uniformly random completion of $y$ makes $T$ equal to 1. Observe that we can sample $\mathbf{x}\in\{0,1\}^n$ uniformly at random among completions of $y$ as follows: independently for all $i\in \{1, \ldots, n\}$, pick $\mathbf{x}_i$ uniformly at random if $y_i = \bot$ or set $\mathbf{x}_i = y_i$ otherwise.

We say that $S\subseteq \{1, 2, \ldots, n\}$ is a $\delta$-sufficient reason for $x\in\{0, 1\}^n$ under $T$ if the following holds:
\begin{equation}
\label{suff}
\frac{\Big|\{z\in\{0, 1\}^n\mid z_i = x_i \text{ for every } i\in S \text{ and } T(z) = T(x)\Big|}{\Big|\{z\in\{0, 1\}^n\mid z_i = x_i \text{ for every } i\in S\Big|}\ge \delta.
\end{equation}
Observe that this definition can also be formulated in probabilistic terms. Namely, let $\mathbf{z}$ be the following random variable: independently for all $i\in\{1, \ldots, n\}$, pick $\mathbf{z}_i\in\{0,1\}$ uniformly at random if $i\notin S$ and set $\mathbf{z}_i = x_i$ if $i\in S$. Then \eqref{suff} can be understood as  $\Pr[T(\mathbf{z})  = T(x)]\ge \delta$.

\begin{proposition}[Hoeffding's inequality, \cite{hoeffding1963probability}]
\label{hoe}
Let $X_1, X_2, \ldots, X_n$ be independent Bernoulli random variables. Denote $S_n = X_1 +\ldots + X_n$. Then for every $\delta > 0$ we have:
\[\Pr\left[\frac{S_n}{n} \ge \frac{\mathbb{E}S_n}{n} + \delta\right]\le \exp\{-2\delta^2 n\}.\]
\end{proposition}

\section{Intermediate problem}
\label{inter}
To establish 
Theorems \ref{minimum}, we introduce the following intermediate problem. Given a decision tree $T$, we want to fix some variables of $T$ to 1 such that the probability that $T$ equals 1 is maximized (assuming that unfixed variables are chosen independently uniformly at random). It is irrelevant how many variables we fix to 1, the point is that we cannot fix anything to $0$. In principle, we can fix all variables to 1, so the problem is interesting only if $T(1, 1, \ldots, 1) = 0$.

\begin{theorem}
\label{fix}
    Unless SAT can be solved in quasi-polynomial time, for every $\kappa > 0$ there exists no polynomial-time algorithm that, given a decision tree $T$ over $n$ variables, distinguishes between the following two cases:
    \begin{itemize}
        \item there exists a partial input $y\in \{1, \bot\}^n$ such that $T(y) \ge 1 - \kappa$;
        \item there exists no partial input $y\in\{1, \bot\}^n$ such that $T(y) \ge \kappa$.
    \end{itemize}
\end{theorem}
The proof of this result is given in the next section. We first derive Theorem \ref{minimum} from it.

\begin{proof}[Theorem \ref{fix} $\implies$ Theorem \ref{minimum}]
   For every $\varepsilon > 0$, we provide a polynomial-time reduction from the problem from Theorem \ref{fix} with $\kappa = \varepsilon/2$ to the problem from Theorem \ref{minimum} with parameter $\varepsilon$.
    Define $m = \lceil (n + \log_2(2/\varepsilon))^{1/\varepsilon}\rceil$.
   Given a decision tree $T$ over $n$ variables, consider a decision tree $T_1$ over $n + m$ variables, defined as follows:
   \[T_1(x_1, \ldots, x_{n + m}) = T(x_1, \ldots, x_n)\lor (x_{n+1}\land x_{n+2}\land \ldots \land x_{n+m}).\]
    Note that $T_1$ can be constructed in polynomial time from $T$ (to every leaf of $T$ that outputs 0 attach a tree computing the conjunction $x_{n+1}\land x_{n+2}\land \ldots \land x_{n+m}$). Let $e^k$ denote a Boolean vector, consisting of $k$ ones. It remains to establish two claims.
   \begin{itemize}
       \item If there exists a partial input $y\in\{1, \bot\}^n$ such that $T(y) \ge 1 - \kappa$, then there exists a $(1 - \varepsilon)$-sufficient reason for $e^{n+m}$ under $T_1$ of size at most $(n + m)^\varepsilon$.

       \item If there exists no partial input $y\in\{1, \bot\}^n$ such that $T(y) \ge \kappa$, then there exists no $\varepsilon$-sufficient reason for $e^{n+m}$ under $T_1$ of size less than $(n + m) - (n + m)^\varepsilon$.
   \end{itemize}
   Let us start with the first claim. First, note that $T(e^{n+m}) = 1$. Define $S = \{i\in\{1, 2, \ldots, n\}\mid y_i = 1\}$. The size of $S$ is at most $n = (n^{1/\varepsilon})^\varepsilon \le m^\varepsilon \le (n + m)^\varepsilon$. It remains to show that $S$ is a $(1 - \varepsilon)$-sufficient reason for $e^{n+m}$ under $T_1$. In other words, we have to show that $\Pr[T_1(\mathbf{z}) = 1] \ge 1 - \varepsilon$, where $\mathbf{z}$ is the following random variable: independently for all $i\in\{1, \ldots, n + m\}$, pick $\mathbf{z}_i$ uniformly at random if $i\notin S$ and set $\mathbf{z}_i = 1$ otherwise. Observe that $\mathbf{z}_1\ldots \mathbf{z}_n$ is a uniformly random completion of $y$. Hence, the probability that $T(\mathbf{z}_1\ldots \mathbf{z}_n) = 1$ is $T(y) \ge 1 - \kappa \ge 1 - \varepsilon$. It remains to notice that $T(\mathbf{z}_1\ldots \mathbf{z}_n) = 1 \implies T_1(\mathbf{z}) = 1$.

   Now, let us show the second claim. Assume that $S\subseteq \{1, \ldots, n + m\}$ is an $\varepsilon$-sufficient reason for $e^{n + m}$ under $T_1$. Our goal is to show that the size of $S$ is at least $(n + m) - (n + m)^\varepsilon$. Again, let $\mathbf{z}$ be the following random variable: independently for all $i\in\{1, \ldots, n + m\}$, pick $\mathbf{z}_i$ uniformly at random if $i\notin S$ and set $\mathbf{z}_i = e^{n +m}_i = 1$ otherwise. The fact that $S$ is an $\varepsilon$-sufficient reason for $e^{n+m}$ under $T_1$ means that $\Pr[T_1(\mathbf{z}) = 1]\ge \varepsilon$. By definition of $T_1$ and  by the union bound we have:
   \begin{equation}
       \label{union}
   \varepsilon\le\Pr[T_1(\mathbf{z}) = 1] \le \Pr[T(\mathbf{z}_1\ldots \mathbf{z}_n) = 1] + \Pr[\mathbf{z}_{n + 1}\land\ldots \land\mathbf{z}_{n+m} = 1]
   \end{equation}
   Observe that $\mathbf{z}_1\ldots \mathbf{z}_n$ is a uniformly random completion of $y\in\{1,\bot\}^n$, defined by 
   \[y_i= \begin{cases}1 & i\in S, \\ \bot & i\notin S,\end{cases} \qquad i\in\{1, \ldots, n\}.\]
   Hence, $\Pr[T(\mathbf{z}_1\ldots \mathbf{z}_n) = 1] = T(y) < \kappa = \varepsilon/2$. Thus, from \eqref{union} we get $\Pr[\mathbf{z}_{n + 1}\land\ldots \land\mathbf{z}_{n+m} = 1] > \varepsilon/2$. Since coordinates of $\mathbf{z}$ are independent, we have $\Pr[\mathbf{z}_{n + 1}\land\ldots \land\mathbf{z}_{n+m} = 1]  = \Pr[\mathbf{z}_{n + 1} = 1] \cdot \ldots \cdot \Pr[\mathbf{z}_{n + m} = 1]$. By definition, $\Pr[\mathbf{z}_{i} = 1] = 1/2$ if $i\notin S$ and  $\Pr[\mathbf{z}_{i} = 1] = 1$ if $i\in S$. This means that $\varepsilon/2 < \Pr[\mathbf{z}_{n + 1}\land\ldots \land\mathbf{z}_{n+m} = 1] = 2^{-|\{n+1, \ldots, n +m\}\setminus S|}$. This gives us $|\{n+1, \ldots, n +m\}\setminus S| < \log_2(2/\varepsilon)$. In particular, we obtain that $|S| \ge m - \log_2(2/\varepsilon)$. It remains to show that $m - \log_2(2/\varepsilon) \ge (n + m) - (n + m)^\varepsilon$, or, equivalently, $(n + m)^\varepsilon \ge n + \log_2(2/\varepsilon)$. Indeed, $(n + m)^\varepsilon \ge m^\varepsilon = \lceil (n + \log_2(2/\varepsilon))^{1/\varepsilon}\rceil^\varepsilon \ge n + \log_2(2/\varepsilon)$.
\end{proof}

Implicitly, W{\"a}ldchen et al.~\cite{waldchen2021computational} establish Theorem \ref{fix} when $T$ is a Boolean circuit. Then the same reduction as above establishes their hardness-of-approximation result for Boolean circuits.

Here is a sketch of their argument. The reduction is from the satisfiability problem for CNFs. Given a CNF $C$, we replace every its variable by a sufficiently long conjuction of fresh variables. Let the resulting circuit be $T$. Assume that $T$ is over $N$ variables. If $C$ is unsatisfiable, then $T(y) = 0$ for every $y\in\{1,\bot\}^N$. Assume now  that $C$ is satisfiable. We take any satisfying assignment $\alpha$ of $C$ and convert it into $y\in\{1,\bot\}^N$ such that $T(y)$ is very close to 1. Namely, if a variable of $C$ is set to 1 in $\alpha$, we fix all variables in the corresponding conjuction to 1s. Otherwise, we leave all variables in this conjuction undefined. With high probability, the value of each conjuction will be equal to the value of the corresponding variable in  $\alpha$. Hence, with high probability, the value of $T$ on a random completion of $y$ will be equal to $C(\alpha) = 1$.

We know present an argument for decision trees. The main challenge is that satisfiability problem for decision trees is polynomial-time solvable. Hence, we have to establish the hardness of this problem based on different ideas.

\section{Proof of Theorem \ref{fix}}
We reduce from the 1-in-$k$ exact hitting set problem (1-in-E$k$HS) problem. It is a version of the SAT problem where each clause contains exactly $k$ variables (without negations) and a clause is satisfied if and only if there is exactly one variable in it which is set to 1. We use the following result of Guruswami and Trevisan~\cite[Lemma 13]{guruswami2005complexity}:

\begin{theorem}    For every $\delta>0$ there exists $k$ such that it is NP-hard to distinguish satisfiable instances of 1-in-E$k$HS from instances of 1-in-E$k$HS for which it is impossible to satisfy more than a fraction $(1/e  + \delta)$ of clauses.
\end{theorem}

We fix any $k$ such that it is NP-hard to distinguish satisfiable instances of 1-in-E$k$HS from instances of 1-in-E$k$HS for which it is impossible to satisfy more than a fraction $1/2$ of clauses. Now, fix $\kappa > 0$. To establish Theorem \ref{fix}, we take any instance $\varphi$ of 1-in-E$k$HS with $n$ variables and $m$ clauses and in \emph{quasi-polynomial} time construct a decision tree $T$ over $N$ variables such that
\begin{itemize}
    \item if $\varphi$ is satisfiable, then there exists $y\in\{1,\bot\}^N$ such that $T(y)\ge 1 - \kappa$;
    \item if it is impossible to satisfy more than a fraction $1/2$ of clauses of $\varphi$, then there exists no $y\in\{1,\bot\}^N$ such that $T(y)\ge\kappa$.
\end{itemize}

Let $l$ be the smallest integral number such that $m\le 2^{2l}$. Note that $l = \Theta(\log m)$. The most technical part of the proof is to construct in polynomial time an $O(\log m)$-depth decision tree $L$ over $n + 2l + 2$ variables such that for some absolute constant $c > 0$ we have:
\begin{itemize}
    \item if $\varphi$ is satisfiable, then there exists $y\in\{1,\bot\}^{n+2l + 2}$ such that $L(y)\ge 7/8$;
    \item if it is impossible to satisfy more than a fraction $1/2$ of clauses of $\varphi$, then there for every $y\in\{1,\bot\}^{n+2l +2}$ we have $L(y)\le 7/8 - \delta$, where $\delta = \frac{c}{\sqrt{\ln(m)}}$.
\end{itemize}
We first describe how to ``increase the gap''  from  ($7/8 - \delta$ vs.~$7/8$) to ($\kappa$ vs.~$1 - \kappa$). The construction of $L$ is given afterwards.

\bigskip

\textbf{Finishing the proof modulo the construction of $L$.} We  define a decision tree $T$ as
\[K = \frac{2\ln(2/\kappa)}{\delta^2} = O(\log m)\] ``independent copies'' of $L$. More specifically, we let $T$ be over $N = K(n + 2l + 2)$ variables. First, $T$ runs $L$ on the first $n + 2l + 2$ variables, then on the second $n + 2l +2$ variables, and so on (in total, $K$ runs). In the end, it outputs 1 if and only if for at least a fraction $7/8 - \delta/2$ of runs the output of $L$ was 1.

Assume first that $\varphi$ is satisfiable. Then, by definition of $L$,
there exists $y\in\{1,\bot\}^{n + 2l + 1}$ such that $L(y)\ge 7/8$. Repeat it $K$ times to obtain a partial input $Y\in\{1,\bot\}^{N}$ for $T$. We claim that $T(Y)\ge 1 - \kappa$. Indeed, a random completion of $Y$ can be generated as $K$ independent samples of a random completion of $y$.  The  tree $T$ outputs 1 if and only if $L$ outputs 1 for at least a fraction $7/8 - \delta/2$ of the samples. On the other hand, the probability that $L$ outputs 1 on one of the samples is $L(y)\ge 7/8$. Hence, the average fraction of samples on which $L$ outputs 1 is at least $7/8$. By Hoeffding's inequality (Proposition \ref{hoe}), the probability that $T$ outputs 1 is at least $1 - \exp\{-2(\delta/2)^2 \cdot K\} = 1 - \kappa/2 > 1 - \kappa$.  

Now, assume that it is impossible to satisfy more than a fraction $1/2$ of clauses of $\varphi$. Then
$L(y) \le 7/8 - \delta$ for all $y\in\{1,\bot\}^{n + 2l + 2}$. We claim that $T(Y) < \kappa$ for every $Y\in\{1,\bot\}^{N}$. Let $y_1$ be a restriction of $Y$ to the first $n + 2l + 2$ variables, $y_2$ be a restriction of $Y$ to the second $n + 2l + 2$ variables, and so on. Generating a random completion of $Y$ is the same as independently sampling random completions of $y_1, y_2,\ldots, y_K$. For each $i\in\{1, \ldots, K\}$, the probability that $L$ outputs 1 on a random completion of $y_i$ is $L(y_i) \le 7/8 - \delta$.  Hence, the average fraction of $i\in\{1, \ldots, K\}$ such that $L$ outputs 1 on a random completion of $y_i$ is at most $7/8 - \delta$. In turn, $T(Y)$ is  the probability that the fraction of $i\in\{1, \ldots, K\}$ such that $L$ outputs 1 on a random completion of $y_i$ is at least $7/8$. Hence, by  Hoeffding's inequality, $T(Y)$ is bounded by  $\exp\{-2(\delta/2)^2 \cdot K\} = \kappa/2 < \kappa$.

The depth of $T$ is $\text{(depth of $L$)}\times K = O(\log^2 m)$. Hence, the size of $T$ is quasi-polynomial. The time to  construct $T$ from $L$ is also quasi-polynomial because it is polynomial in the size of $T$.

\bigskip

\textbf{Construction of $L$.} Let $x_1,\ldots, x_n$ denote variables of $\varphi$. The variables of $L$ will be denoted by $x_1, \ldots, x_n, y_1, \ldots, y_{2l + 1}, z$.

First, $L$ asks the values of $y_1, \ldots, y_{2l + 1}$. Now, we will call a binary word ``fat'' if it has more 1s than 0s, and ```thin'' otherwise.  If $y_1\ldots  y_{2l +1}$ is thin, $L$ outputs 1. To describe what $L$ does when $y_1\ldots  y_{2k +1}$ is fat, we first notice that in $\{0, 1\}^{2k + 1}$ exactly half of the words are fat and exactly half are thin. In particular, the number of fat words is $2^{2l}$. We fix any surjective mapping from the set of fat words in $\{0, 1\}^{2l + 1}$ to the set of clauses of $\varphi$. Such mapping exists because $m\le 2^{2l}$ by our choice of $l$. 

So, when $y_1\ldots  y_{2l +1}$ is fat, $L$ takes the  clause $C$, assigned to $y_1\ldots  y_{2l +1}$, and runs a decision tree $L_C$, defined as follows.
If $C$ is over variables $x_{i_1}, \ldots, x_{i_k}$, then $L_C$ is over  $x_{i_1}, \ldots, x_{i_k}$ and $z$.    The tree $L_C$ asks the values of all its variables. It outputs 1 if and only if either there is exactly one 0 among $x_{i_1}, \ldots, x_{i_k}$ or ($x_{i_1} = \ldots = x_{i_k} = 1$ and $z = 0$).

The depth of $L$ is $2l + 1 + k + 1 = O(l) = O(\log m)$, and it can be constructed from $\varphi$ in polynomial time.

\begin{lemma}
\label{cases}
    Consider any partial input $p$ to $L_C$ in which no coordinate is equal to 0 (i.e, we can only have variables that are fixed to 1 and undefined variables).  We call $p$ good if $z$ is undefined  in $p$ and there is also    exactly one undefined variable among $x_{i_1}, \ldots, x_{i_k}$ in $p$. Otherwise, we call $p$ bad. 

Then,
    if $p$ is good, we have $L_C(p) = 3/4$, and if $p$ is bad, we have $L_C(p) \le 5/8$.
    
\end{lemma}
\begin{proof}
Assume first that $p$ is good. Let  $x_{i_j}$ be a variable among $x_{i_1}, \ldots, x_{i_k}$ which is undefined  in $p$. 
Consider a random completion of $p$.
If $x_{i_j}$ is $0$ in this completion, then there  is exactly one 0  among $x_{i_1}, \ldots, x_{i_k}$, and hence $L_C$ outputs 1. If $x_{i_j}$ is 1, $L_C$ outputs 1 with probability $1/2$, depending on whether $z = 0$. Overall, we get $L_C(p) = (1/2)\cdot 1 + (1/2)\cdot (1/2) = 3/4$.

    Now, consider the case when $p$ is bad. 
    Assume first that there are $t\ge 2$ undefined variables among $x_{i_1}, \ldots, x_{i_k}$ in $p$. Consider a random completion of $p$. By definition, $L_C$ outputs 1 on it only in the following two cases: (a) there is exactly 1 undefined variable among $x_{i_1}, \ldots, x_{i_k}$ which is equal to 0 in our random completion; (b) all undefined variables  among $x_{i_1}, \ldots, x_{i_k}$ are equal to 1 and $z = 0$ in our random completion. The probability of (a) is $t2^{-t}$. The probability of (b) is $2^{-t - 1}$ if $z$ is undefined and $0$ otherwise. Overall, we get $L_C(p) \le t2^{-t} + 2^{-t - 1}$. This expression decreases in $t$, and for $t = 2$ it is equal to $5/8$.

    Now, assume that the number of undefined variables among $x_{i_1}, \ldots, x_{i_k}$ is at most 1. If it is 1, then $z$ has to be fixed to 1 in $p$ (otherwise, $p$ is good). Then $L_C$ can only output 1 if the unique  undefined variable among $x_{i_1}, \ldots, x_{i_k}$ is 0, and this happens with probability $1/2$ in a random completion of $p$. That is, in this case, $L_C(p) = 1/2$. Now, if all variables $x_{i_1}, \ldots, x_{i_k}$ are fixed to 1s in $p$, then $L_C$ outputs 1 only if $z = 0$. This either has probability $1/2$ (if $z$ is undefined) or $0$ (if $z$ is fixed to 1). Hence, in this case, $L_C(p) \le 1/2$.
\end{proof}

Assume first that $\varphi$ is satisfiable. Fix any satisfying assignment $\alpha$ to $\varphi$. We construct a partial input $p\in\{1, \bot\}^{n + 2l + 2}$ to $L$ such that $L(p) = 7/8$. We define  $p$ on $x_1, \ldots, x_n$  as follows:
\begin{align}
\label{0rule}
    \mbox{$x_i = 0$ in $\alpha$} &\implies \mbox{$x_i = 1$  in $p$}\\
    \label{1rule}
    \mbox{$x_i = 1$ in $\alpha$} &\implies \mbox{$x_i = \bot$ in $p$}
\end{align}
Variables $y_1, \ldots, y_{2l+1}, z$ are undefined in $p$.
Since $\alpha$ is a satisfying assignment to $\varphi$, for every clause $C$ there is exactly one variable in $C$ which is undefined in $p$. This means that the restriction of $p$ to variables of $L_C$ is good in the sense of Lemma \ref{cases}, for every clause of $C$. Thus, conditioned on the event that $y_1\ldots y_{2l +1}$ is fat, $L$ outputs 1 with probability $3/4$ on a random completion of $p$. Now, variables $y_1, \ldots, y_{2l + 1}$ are undefined in $p$, which means that they are all sampled independently uniformly at random in our completion. The probability of the event  ``$y_1\ldots y_{2l +1}$ is fat'' is $1/2$, which contributes $(1/2)\cdot (3/4) = 3/8$ to $L(p)$. In turn,  conditioned on the event  ``$y_1\ldots y_{2l +1}$ is thin'' (which also happens with probability $1/2$), our tree always outputs 1. Overall, we get $L(p) = 3/8 + 1/2 = 7/8$.

We now show that if it is impossible to satisfy more than a half of clauses of $\varphi$, then for every $p\in\{1, \bot\}^{n + 2l + 2}$ we have $L(p) \le 7/8 - \delta$, where $\delta = c/\sqrt{\ln(m)}$ and $c > 0$ is some absolute constant. Again, we do so by taking a random completion of $p$ and showing that the probability that $L$ outputs 1 on it is at most $7/8 - \delta$.

Assume first that $y_1, \ldots, y_{2l+1}$ are all undefined in $p$. We turn $p$ into an assignment $\alpha$ to $\varphi$ by reversing (\ref{0rule}--\ref{1rule}):
\begin{align*}
     \mbox{$x_i = 1$ in $p$} &\implies \mbox{$x_i = 0$ in $\alpha$} \\
     \mbox{$x_i = \bot$ in $p$} &\implies  \mbox{$x_i = 1$ in $\alpha$}
\end{align*}
Observe that if $\alpha$ does not satisfy a clause $C$, then the restriction of $p$ to variables of $L_C$ is bad in the sense of Lemma \ref{cases}.
Now, the number of unsatisfied clauses is at least $m/2$. Each of these clauses is assigned to some fat $y_1\ldots y_{2l +1}$. I.e., for at least $m/2$ fat $y_1\ldots y_{2k +1}$, the probability that $L$ outputs 1 (conditioned on this fixation of $y_1\ldots y_{2k +1}$) is at most $5/8$. For any other fat  $y_1\ldots y_{2k +1}$, this conditional probability  is at most $3/4$.
Overall, the contribution of fat $y_1\ldots y_{2k +1}$ to $L(p)$ does not exceed 
\[\frac{m/2}{2^{2l + 1}}\cdot (5/8) +\frac{2^{2l} - m/2}{2^{2l + 1}}\cdot (3/4), \]
and the contribution of thin  $y_1\ldots y_{2k +1}$ to $L(p)$ is, as before, $1/2$. Overall, we get
\[L(p) \le \frac{m/2}{2^{2l + 1}}\cdot (5/8) +\frac{2^{2l} - m/2}{2^{2l + 1}}\cdot (3/4) + 1/2 = 7/8 - \frac{m}{2^{2l + 5}}.\]
It remains to recall that $l$ was chosen as the smallest integral number such that $m\le 2^{2l}$. This means that $m\ge 2^{2(l-1)}$. This gives us $L(p) \le 7/8 - 1/128 < 7/8 - \delta$.

It remains to consider the case when in $p$ at least one variable among $y_1, \ldots, y_{2l + 1}$ is fixed to 1. As before, the probability that $L$ outputs 1, conditioned on the event ``$y_1, \ldots, y_{2k + 1}$ is fat'', is at most $3/4$. And this probability is 1 conditioned on the event ``$y_1, \ldots, y_{2k + 1}$ is thin''. Now, however, since at least one variable among $y_1, \ldots, y_{2k + 1}$ is fixed to 1, the probability to be thin is slightly less than $1/2$, which makes $L(p)$ slightly less than $7/8$. More specifically, 
\begin{align*}
L(p) &\le \Pr[y_1\ldots y_{2l + 1} \mbox{ is thin}] + (3/4)\Pr[y_1\ldots y_{2l + 1} \mbox{ is fat}] \\
&= 1 - (1/4)\Pr[y_1\ldots y_{2l + 1} \mbox{ is fat}], 
\end{align*}
where $y_1,\ldots, y_{2l +1}$ are chosen from a random completion of $p$. That is, now our goal is to lower bound $\Pr[y_1\ldots y_{2l + 1} \mbox{ is fat}]$, assuming that at least one of $y_1,\ldots, y_{2l + 1}$ is fixed to 1, and the rest of them are chosen independently uniformly at random. The more variables we fix to 1, the larger becomes $\Pr[y_1\ldots y_{2l + 1} \mbox{ is fat}]$. Hence, w.l.o.g.~we may assume that exactly one variable is fixed to 1, say, $y_{2l +1}$. Then $y_1\ldots y_{2l + 1}$ is fat if and only if there are at least $l$ ones among $y_1, \ldots, y_{2l}$. Thus, we get
\begin{align*}
   \Pr[y_1\ldots y_{2l + 1} \mbox{ is fat}]&\ge \frac{\binom{2l}{2l} + \ldots + \binom{2l}{l+1} +\binom{2l}{l} }{2^{2l}} \\
   &= \frac{\frac{\binom{2l}{2l} + \binom{2l}{0}}{2} + \ldots + \frac{\binom{2l}{l+1} + \binom{2l}{l - 1}}{2} +\frac{\binom{2l}{l}+\binom{2l}{l}}{2}}{2^{2l}}\\
   &= \frac{\binom{2l}{2l} + \ldots + \binom{2l}{0}}{2\cdot 2^{2l}} + \frac{\binom{2l}{l}}{2\cdot 2^{2l}} = \frac{1}{2} + \Omega(1/\sqrt{l})\\
   &= \frac{1}{2} + \Omega(1/\sqrt{\ln m}).
\end{align*}
From this, we get $L(p) \le 1 - (1/4)\cdot (1/2 + \Omega(1/\sqrt{\ln m})) = 7/8 - \Omega(1/\sqrt{\ln m})$, as required.

\bibliographystyle{acm}
\bibliography{ref}

\begin{thebibliography}{10}

\bibitem{arenas2022computing}
{\sc Arenas, M., Barcel{\'o}, P., Romero~Orth, M., and Subercaseaux, B.}
\newblock On computing probabilistic explanations for decision trees.
\newblock {\em Advances in Neural Information Processing Systems 35\/} (2022),
  28695--28707.

\bibitem{barcelo2020model}
{\sc Barcel{\'o}, P., Monet, M., P{\'e}rez, J., and Subercaseaux, B.}
\newblock Model interpretability through the lens of computational complexity.
\newblock {\em Advances in neural information processing systems 33\/} (2020),
  15487--15498.

\bibitem{guruswami2005complexity}
{\sc Guruswami, V., and Trevisan, L.}
\newblock The complexity of making unique choices: approximating 1-in-k sat.
\newblock In {\em Proceedings of the 8th international workshop on
  Approximation, Randomization and Combinatorial Optimization Problems, and
  Proceedings of the 9th international conference on Randamization and
  Computation: algorithms and techniques\/} (2005), pp.~99--110.

\bibitem{hoeffding1963probability}
{\sc Hoeffding, W.}
\newblock Probability inequalities for sums of bounded random variables.
\newblock {\em Journal of the American Statistical Association 58}, 301 (1963),
  13--30.

\bibitem{ignatiev2021sat}
{\sc Ignatiev, A., and Marques-Silva, J.}
\newblock Sat-based rigorous explanations for decision lists.
\newblock In {\em Theory and Applications of Satisfiability Testing--SAT 2021:
  24th International Conference, Barcelona, Spain, July 5-9, 2021, Proceedings
  24\/} (2021), Springer, pp.~251--269.

\bibitem{lundberg2017unified}
{\sc Lundberg, S.~M., and Lee, S.-I.}
\newblock A unified approach to interpreting model predictions.
\newblock {\em Advances in neural information processing systems 30\/} (2017).

\bibitem{ribeiro2018anchors}
{\sc Ribeiro, M.~T., Singh, S., and Guestrin, C.}
\newblock Anchors: High-precision model-agnostic explanations.
\newblock In {\em Proceedings of the AAAI conference on artificial
  intelligence\/} (2018), vol.~32.

\bibitem{shih2018symbolic}
{\sc Shih, A., Choi, A., and Darwiche, A.}
\newblock A symbolic approach to explaining bayesian network classifiers.
\newblock In {\em Proceedings of the 27th International Joint Conference on
  Artificial Intelligence\/} (2018), pp.~5103--5111.

\bibitem{waldchen2021computational}
{\sc W{\"a}ldchen, S., Macdonald, J., Hauch, S., and Kutyniok, G.}
\newblock The computational complexity of understanding binary classifier
  decisions.
\newblock {\em Journal of Artificial Intelligence Research 70\/} (2021),
  351--387.

\bibitem{yan2021if}
{\sc Yan, T., and Procaccia, A.~D.}
\newblock If you like shapley then you’ll love the core.
\newblock In {\em Proceedings of the AAAI Conference on Artificial
  Intelligence\/} (2021), vol.~35, pp.~5751--5759.

\end{thebibliography}

\end{document}